\documentclass{amsart}

\usepackage[american]{babel}

\usepackage[letterpaper,top=2cm,bottom=2cm,left=3cm,right=3cm,marginparwidth=1.75cm]{geometry}

\usepackage{comment}

\usepackage{enumerate}

\usepackage{amsmath}
\usepackage{amsthm}
\usepackage{graphicx}
\usepackage[colorlinks=true, allcolors=red]{hyperref}
\usepackage{cite}

\usepackage[american]{babel}	
\usepackage{csquotes}	
\usepackage{graphicx}			
\usepackage{epstopdf}
\usepackage{subfig}				
\usepackage{marvosym}		
\usepackage{amssymb}
\usepackage{mathtools}
\usepackage{xcolor}

\DeclareMathOperator*{\argmin}{arg\,min}

\newtheorem{theorem}{Theorem}
\newtheorem{proposition}{Proposition}
\newtheorem{corollary}{Corollary}
\newtheorem{observation}{Observation}
\newtheorem{example}{Example}
\newtheorem{conjecture}{Conjecture}

\title[Correctness of MP in the NNI neighborhood]{On the correctness of Maximum Parsimony for data with few substitutions in the NNI neighborhood of phylogenetic trees}

\author{Mareike Fischer}

\address{Institute of Mathematics and Computer Science, Greifswald University, Greifswald, Germany} \email{mareike.fischer@uni-greifswald.de, email@mareikefischer.de}

\begin{document}

\begin{abstract} 
Estimating phylogenetic trees, which depict the relationships between different species, from aligned sequence data (such as DNA, RNA, or proteins) is one of the main aims of evolutionary biology. However, tree reconstruction criteria like maximum parsimony do not necessarily lead to unique trees and in some cases even fail to recognize the \enquote{correct} tree (i.e., the tree on which the data was generated). On the other hand, a recent study has shown that for an alignment containing precisely those binary characters (sites) which require up to two substitutions on a given tree, this tree will be the unique maximum parsimony tree.
 
 It is the aim of the present paper to generalize this recent result in the following sense: We show that for a tree $T$ with $n$ leaves, as long as $k<\frac{n}{8}+\frac{11}{9}-\frac{1}{18}\sqrt{9\cdot \left(\frac{n}{4}\right)^2+16}$ (or, equivalently, $n>9 k-11+\sqrt{9k^2-22 k+17} $, which in particular holds for all $n\geq 12k$), the maximum parsimony tree for the alignment containing all binary characters which require (up to or precisely) $k$ substitutions on $T$ will be unique in the NNI neighborhood of $T$ and it will coincide with $T$, too. In other words, within the NNI neighborhood of $T$, $T$ is the unique most parsimonious tree for the said alignment. This partially answers a recently published conjecture affirmatively. Additionally, we show that for $n\geq 8$ and for $k$ being in the order of $\frac{n}{2}$, there is always a pair of phylogenetic trees $T$ and $T'$ which are NNI neighbors, but for which the alignment of characters requiring precisely $k$ substitutions each on $T$ in total requires fewer substitutions on $T'$.  
\smallskip \newline
\noindent \textbf{MSC identifiers.} 05C05, 05C90, 92B05
\end{abstract}

\maketitle

\section{Introduction}

One of the main aims of mathematical phylogenetics is the reconstruction of an evolutionary relationship tree, also often called a phylogenetic tree, for a given species set $X$ based on some given data. Typically, the data are provided in the form of aligned sequence data, like DNA, RNA, proteins or binary characteristics (e.g., the absence or presence of certain morphological traits in species). The columns of such alignments are often referred to as \emph{characters} or \emph{sites}. Mathematically speaking, a character is merely a function assigning each species under consideration (which are represented by the leaves of the tree) a state -- which may be a letter of the four-letter DNA alphabet or, as is often the case with morphological data, of a binary alphabet. An alignment is then just a finite sequence of such characters. 

 While there are various different tree reconstruction methods available \cite{Semple2003,Yang2006,Felsenstein2004}, methods based on the famous maximum parsimony (MP) principle are amongst the best-known. The underlying idea of this principle is simply to find the tree with the smallest number of required substitutions or state changes along the edges of the tree. The implicit assumption of the parsimony principle is thus closely related to the well-known principle of Occam's razor \cite{Ariew1976,Sober2015}, which seeks to find the simplest explanation: As evolutionary state changes (substitutions) on a species level are rare, MP's aim is to find a tree that minimizes the number of such changes. For a given character $f$ and a given tree $T$, the minimum number of substitutions needed to explain $f$ on $T$ is often referred to as \emph{parsimony score} of $f$ on $T$.
 
 While there is some criticism concerning MP \cite{Felsenstein1978}, methods based on this criterion often outperform other methods in settings in which the underlying evolutionary substitution model is non-homogeneous \cite{Kolaczkowski2004}. It is also often assumed that MP is able to identify the \enquote{correct} tree (i.e., the one that has generated the data) whenever the number of substitutions is relatively small \cite{Sourdis1988}. 
 
 This biological context has recently inspired several mathematical publications. In particular, the findings of \cite{pablo}, in which for a given tree $T$ the alignment $A_k(T)$ consisting of all binary characters of parsimony score $k$ on $T$ was analyzed, led to the conjecture that whenever $n\geq 4k+1$ (or, equivalently, whenever $k\leq \frac{n-1}{4}$), $T$ is the unique maximum parsimony tree for $A_k(T)$, i.e., MP will recover $T$ uniquely:

\begin{conjecture}[Conjecture 1 from \cite{Fischer2019}] \label{conj}
Let $T$ be a binary phylogenetic $X$-tree with $|X|=n$. Let $k<\frac{n}{4}$. Then, $T$ is the unique maximum parsimony tree for $A_k(T)$. \end{conjecture}

Conjecture \ref{conj} was published in \cite{Fischer2019} and partially proven there (namely for the cases $k=1$ and $k=2$, where the case $k=1$ is a direct consequence of Buneman's famous theorem \cite{Buneman1971,Semple2003} in mathematical phylogenetics). Moreover, already in \cite{pablo} as well as in \cite{Fischer2019}, an example of a tree $T$ with $n=8$ leaves and $k=2$ can be found for which another tree has a smaller (and thus \enquote{better}) parsimony score than $T$, and this tree is a so-called NNI neighbor of $T$. While NNI will be explained more in-depth in the following section, the idea here is simply that its NNI neighbors differ only slightly from $T$, cf. Example \ref{ex:8taxa}.

\begin{example}\label{ex:8taxa} Figure \ref{fig:8taxaBAD} shows a tree $T$ and one of its so-called NNI neighbors  $T'$, which differs from $T$ only concerning the fact that while $T$ contains edge $e$, which separates leaves 1,2,3 and 4 from 5,6, 7 and 8, $T'$ contains edge $e'$, which separates 1,2,7,8 from 3,4,5,6. It was mentioned in \cite{pablo} and further analyzed in \cite{Fischer2019} that $T'$ has a strictly smaller parsimony score for $A_2(T)$ than $T$.  
\end{example}

\begin{figure}
\includegraphics[scale=0.8]{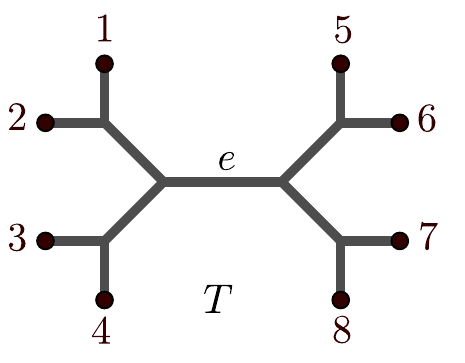} \hspace{0.5cm}
\includegraphics[scale=0.8]{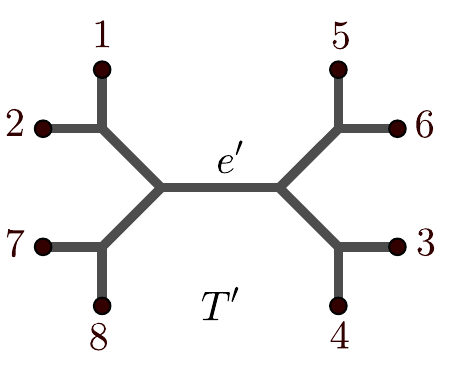}
\caption{Two trees, $T$ and $T'$, which are NNI neighbors (which can be seen by swapping the subtree with leaves 3 and 4 with the one with leaves 7 and 8 around edge $e$ in $T$) and for which the parsimony score of $A_2(T)$ is smaller on $T'$ than on $T$. }
\label{fig:8taxaBAD}
\end{figure}

 Mathematically, it is a natural question to ask if the results from \cite{Fischer2019} concerning the uniqueness of the maximum parsimony tree for $k=2$ can be generalized to $k>2$. It is the main aim of the present paper to show that, at least within the NNI neighborhood of a tree, as long as $k<\frac{n}{8}+\frac{11}{9}-\frac{1}{18}\sqrt{9\cdot \left(\frac{n}{4}\right)^2+16}$, which in particular holds if $n\geq 12k$, this generalization indeed holds. This partially answers Conjecture \ref{conj} affirmatively.

 Moreover, we will re-visit Example \ref{ex:8taxa}, i.e., the setting of a tree that is not a maximum parsimony tree for its own $A_k$-alignment as one of its NNI neighbors has a smaller parsimony score. We will show that this example can be extended in the sense that such examples exist for all  $n\geq 8$ and for $k$ in the magnitude of $\frac{n}{2}$ (for exact values of $k$, see Table \ref{tab:badcases}). We conclude the present paper by discussing our results and highlightling some interesting paths for future research.

\section{Preliminaries}
\subsection{Definitions and basic concepts}

\subsubsection{Phylogenetic trees}
We start with some notation. Recall that a phylogenetic tree $T=(V,E)$ on a species set $X=\{1, \ldots, n\}$ is a connected acyclic graph with vertex set $V$ and edge set $E$ whose leaves are bijectively labeled by $X$. Note that due to this bijective labelling, in a slight abuse of notation, it is common to identify $X$ with $V^1:=\{v \in V: deg_T(v)\leq 1\}$, i.e., with the leaves of $T$, and we will follow this convention in the present paper. Such a tree $T$ is also often referred to as phylogenetic $X$-tree. It is called \emph{rooted} if it contains one designated root node $\rho$ and \emph{unrooted} otherwise. Moreover, it is called \emph{binary} if all its inner nodes have degree 3, except in the rooted case, in which the root must have degree 2 (except in the special case of $|X|=1$ in which the root is at the same time the only leaf of the tree). Note that we consider two phylogenetic $X$-trees $T=(V,E)$ and $T'=(V',E')$ to be isomorphic, denoted $T\simeq T'$, if there exists a map $f:V\rightarrow V'$  such that $e=\{u,v\} \in E \Longleftrightarrow \{f(u),f(v)\} \in E'$ and with the additional property that $f(x)=x$ for all $x \in X$. In the rooted case, we also must have $f(\rho)=\rho'$, where $\rho$ and $\rho'$ denote the roots of $T$ and $T'$, respectively. In other words, $f$ is a graph isomorphism which preserves the leaf labelling, and -- if applicable -- the root position.

Throughout this paper,  whenever we refer to a tree $T$, unless stated otherwise, we always mean an unrooted binary phylogenetic $X$-tree. Whenever we mean a rooted tree, we explicitly state this. However, note that sometimes it is useful to root an unrooted tree by adding an extra root node. In particular, for any (unrooted) binary phylogenetic $X$-tree $T$ we call the rooted binary tree $T_e$ a \emph{rooted version of $T$ with respect to $e$} if $e$ is an edge of $T$ and if $T_e$ can be derived from $T$ by subdividing edge $e$ by adding a new degree-2 root vertex $\rho_e$ to $e$. A depiction of this procedure can be found in Figure \ref{fig_rootingofunrootedtree}. 

\begin{figure}
\includegraphics[scale=0.75]{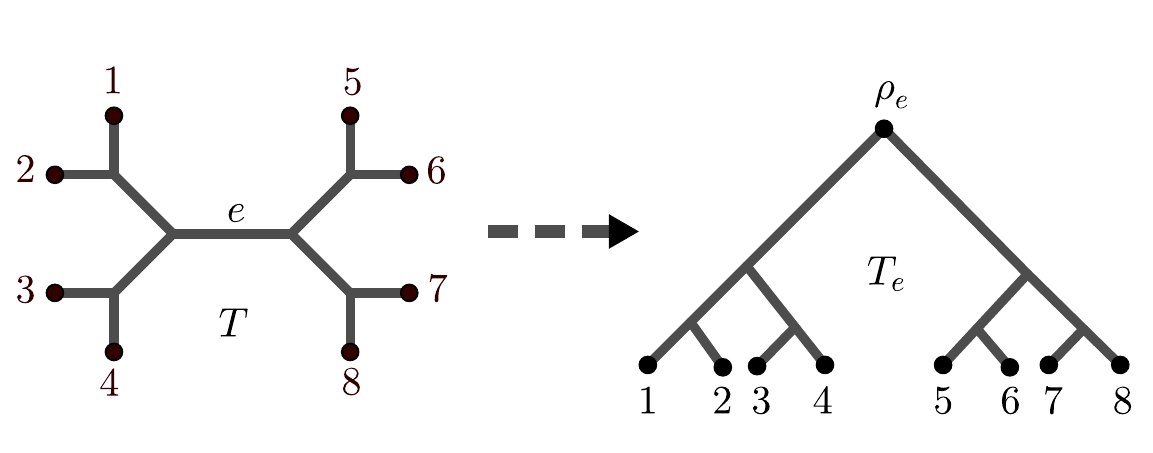} \hspace{0.5cm}
\caption{Subdividing an edge $e$ of a phylogenetic tree $T$ by adding a new degree-2 vertex $\rho_e$ turns $T$ into a rooted phylogenetic tree $T_e$ with root $\rho_e$. }
\label{fig_rootingofunrootedtree}
\end{figure}

Recall that a rooted tree comes with an inherent hierarchy, which implies that the vertices adjacent to a vertex $v$ which are \emph{not} contained on the path from $v$ to the root $\rho$ of the given tree are often referred to as \emph{children} of $v$, and $v$ is their \emph{parent}.

\subsubsection{Characters and alignments}
Now that we have characterized the objects we seek to reconstruct, namely phylogenetic trees, we need to specify the kind of data used to do so. In order to do this, recall that a character $f$ is a function from the taxon set $X$ to a set $\mathcal{C}$ of character states, i.e., $f: X \rightarrow \mathcal{C}$. In this paper, we will only be concerned with \emph{binary characters} (and thus also \emph{binary alignments)}, i.e., without loss of generality $\mathcal{C}= \{a,b\}$. Thus, a binary character $f: X \rightarrow \{a,b\}$ assigns to each leaf of a tree a corresponding state. Note that a finite sequence of characters is also often referred to as \emph{alignment} in biology. While in most biological cases, the order of the characters in an alignment plays an important role, for our purpose it suffices to simply define an alignment as a multiset of characters. Moreover, sometimes two alignments $A$ and $B$ are \emph{concatenated} to form a new alignment $A.B$. This concatenated alignment in our context is simply the union of the two multisets $A$ and $B$.

The most important alignment we will consider, the so-called $A_k$-alignment $A_k(T)$ of a given tree $T$, will be defined at the end of Section \ref{sec_pars}.

\subsubsection{Maximum parsimony and the Fitch algorithm}\label{sec_pars} 
More details on the notions and concepts of this subsection, which play an important role in mathematical phylogenetics, can, for instance, be found in \cite{Semple2003}. In the previous two subsections, we have defined the data (in the form of characters and alignments) and the objects we seek to reconstruct (namely phylogenetic trees). It remains to specify an optimization criterion which we can use to achieve that. In order to do this, we first need to understand what an \emph{extension} of a character is: An \emph{extension} of a binary character $f$ on a phylogenetic $X$-tree $T=(V,E)$ with vertex set $V$ is a map $g: V \rightarrow \{a,b\}$ such that $g(x)=f(x)$ for all $x \in X$. Thus, while $f$ only assigns states to the leaves, $g$ assigns states to all inner vertices of $T$, but it agrees with $f$ on the leaves. Moreover, we call $ch(g) = \vert \{ \{u,v\} \in E, \, g(u) \neq g(v)\} \vert$ the \emph{changing number} or \emph{substitution number} of $g$ on $T$.

Now, the concept we need for tree reconstruction in the present paper is \emph{maximum parsimony (MP)}, which is based on the so-called \emph{parsimony score} $l(f,T)$ of a character $f$ on a tree $T$. Here, $l(f,T) = \min\limits_{g} ch(g,T)$, where the minimum runs over all extensions $g$ of $f$ on $T$. The parsimony score of an alignment $A=\{f_1,\ldots, f_m\}$ of characters is then defined as: $l(A,T)=\sum\limits_{i=1}^m l(f_i,T)$. Moreover, a \emph{maximum parsimony tree}, or $\emph{MP tree}$ for short, of an alignment $A$ is defined as $\argmin_T\{l(A,T)\}$, where the minimum runs over all phylogenetic $X$-trees $T$ defined on the same set of taxa as $A$. In other words, an MP tree of an alignment $A$ is a tree minimizing the parsimony score of $A$ amongst all phylogenetic trees with the taxa of $A$ as leaves.

Note that given a binary tree $T$ and a character $f$, the parsimony score can be efficiently calculated in linear time using the well-known \emph{Fitch algorithm} \cite{Fitch,Hartigan1973}.
This algorithm assigns a set of states to all inner vertices and minimizes the required number of changes. It is based on Fitch's parsimony operation which we explain now. Therefore, let $\mathcal{C}$ be a non-empty finite set of character states and let $A,B \subseteq \mathcal{C}$. Then, Fitch's parsimony operation $*$ is defined by:
$$A*B \coloneqq \begin{cases}
A \cap B, & \text{if } A \cap B \neq \emptyset, \\
A \cup B, & \text{otherwise.}
\end{cases}$$

Using this operation, the Fitch algorithm works as follows. Given a binary phylogenetic tree $T$ and a character $f$, if $T$ is not already rooted, the algorithm first adds a degree-2 root to one of the edges of $T$, cf. Figure \ref{fig_rootingofunrootedtree}. It then proceeds with the rooted tree,  starting with the leaves. In the initialization step, it assigns a set to each leaf containing precisely the state assigned to this leaf by $f$. For instance, if $f(1)=a$, leaf $1$ gets assigned set $S(f,T,1)=\{a\}$. The algorithm then proceeds as follows: In each step, it considers all vertices $v$ whose two children have already been assigned a set, say $A$ and $B$. Then, $v$ is assigned the set $S(f,T,v):=A*B$. This step is continued \enquote{upwards} (i.e., from the leaves to the root) along the tree until the root $\rho$ is assigned a set, which is denoted by $S(f,T,\rho)$. Note that throughout the present paper, when we want to describe the set $S$ assigned to a vertex $v$, by a slight abuse of notation, we often write $S(v)$ instead of $S(f,T,v)$ whenever there is no ambiguity. Ultimately, the parsimony score $l(f,T)$ then simply equals the number of times the Fitch operation had to use the union instead of the intersection \cite{Fitch}. An example for the Fitch algorithm is given by Figure \ref{fig_fitch}.

\begin{figure}
\includegraphics[scale=0.75]{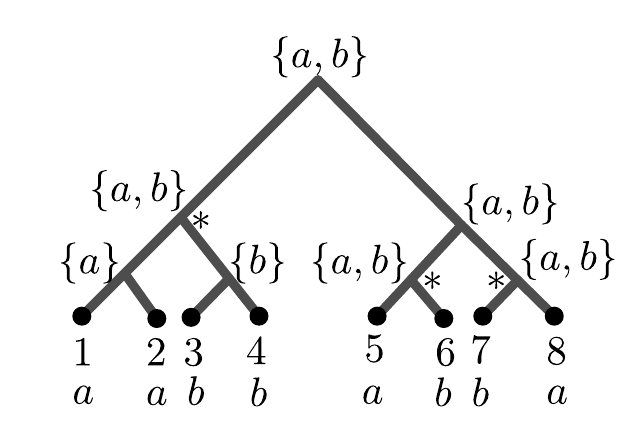} \hspace{0.5cm}
\caption{An example of the Fitch algorithm. Here, character $f$ with $f(1)=a$, $f(2)=a$, $f(3)=b$, $f(4)=b$, $f(5)=a$, $f(6)=b$, $f(7)=b$, $f(8)=a$ is mapped onto the leaves the rooted version $T_e$ of $T$ from Figure \ref{fig_rootingofunrootedtree}. The Fitch algorithm then starts at the leaves and considers their state assignment as sets. It proceeds with those inner vertices whose children have already been assigned a set and applies the Fitch operation. If a union has to be taken, the counter goes up. Note that not all $\{a,b\}$-sets are the result of a union. In the figure, union nodes are marked with an asterix. As there are three such nodes, we have $l(f,T_e)=l(f,T)=3$ and thus $f\in A_3(T)$. }
\label{fig_fitch}
\end{figure}

\par\vspace{0.5cm}
The final and possibly most important concept we wish to introduce in this section is the following: For a given tree $T$, we define $A_k(T)$ to be the set consisting of all binary characters $f$ with $l(f,T)=k$. Following \cite{Fischer2019}, we also refer to $A_k(T)$ as the \emph{alignment induced by $T$ and $k$}.

\subsubsection{Specific notions in the context of NNI}\label{sec:prelimNNI}

When considering a phylogenetic $X$-trees, it is often useful to consider small changes to the tree to see how this affects, for instance, the parsimony score of certain alignments. In this paper, we will therefore consider \emph{nearest neighbor interchange moves} or \emph{NNI moves} for short. An NNI move simply takes an inner edge $e$ of a binary phylogenetic $X$-tree $T$ (i.e., an edge that is \emph{not} incident to a leaf of $T$), and swaps two of the four subtrees of $T$ which we get when deleting the precisely four edges adjacent to $e$ in a way that the resulting tree is not isomorphic to $T$, i.e., in a way that changes the tree. An illustration of NNI moves can be found in Figure \ref{fig_nni}. A tree resulting from $T$ by performing one NNI move is called an \textit{NNI neighbor} of $T$, and all NNI neighbors of $T$ together with $T$ form the \textit{NNI neighborhood} of $T$. Note that this implies that we consider $T$ to belong to its own neighborhood even though $T$ is not its own neighbor.

\begin{figure}
\includegraphics[scale=0.3]{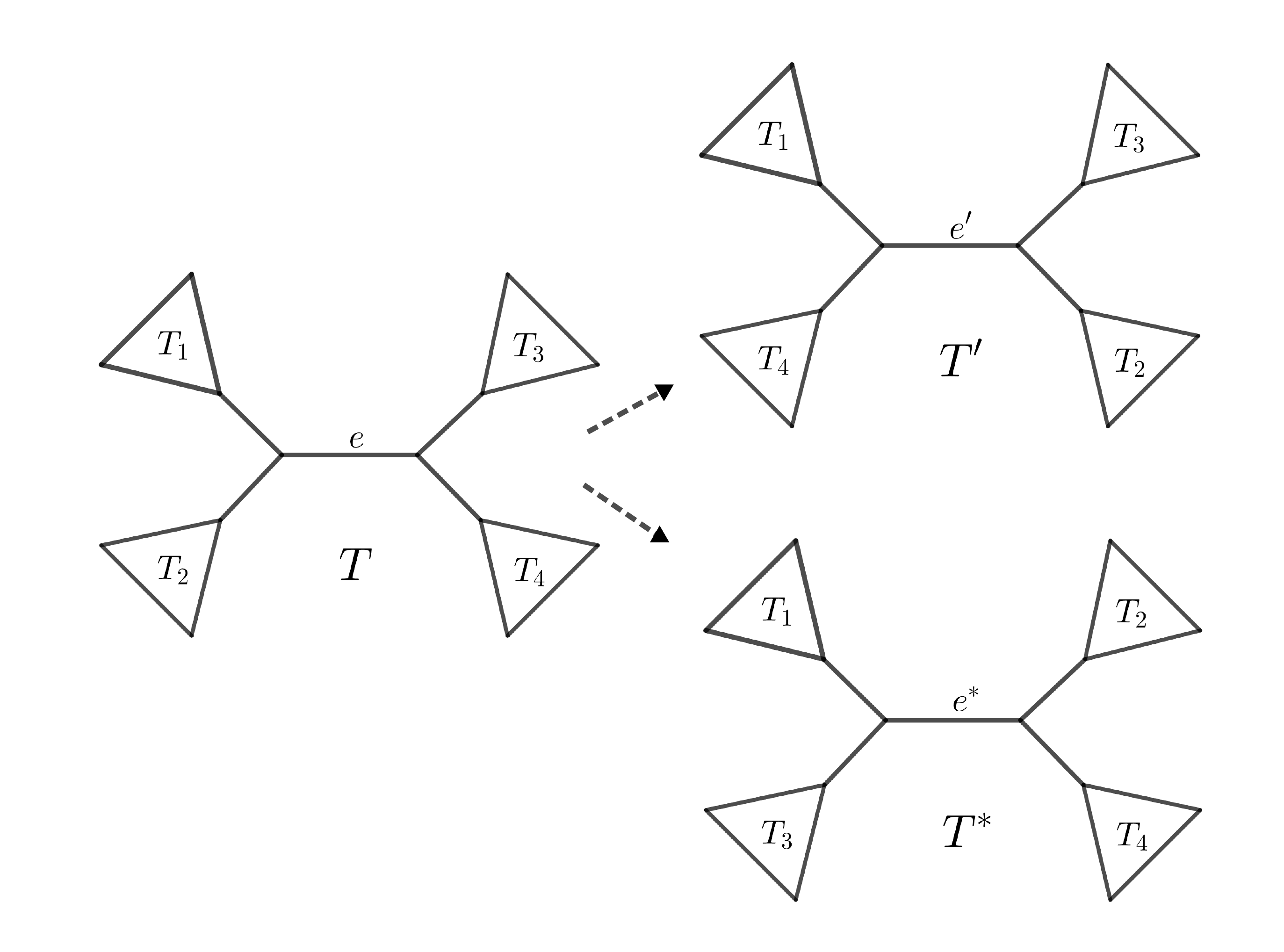} \hspace{0.5cm}
\caption{A tree $T$ and its two NNI neighbors concerning inner edge $e$: Tree $T'$ can be derived from $T$ by swapping subtrees $T_2$ and $T_4$, whereas $T^\ast$ can be derived by swapping subtrees $T_2$ and $T_3$. All subtrees are schematically depicted as triangles. }
\label{fig_nni}
\end{figure}

\par\vspace{0.5cm}
It is the main aim of the present paper to show that if $T'$ is an NNI neighbor of $T$, we have $l(A_k(T),T)<l(A_k(T),T')$ if $k$ is sufficiently small (relative to the leaf number $n$ of $T$). In order to prove this, we now define certain subtrees of $T$ and $T'$ and parameters $\delta$ and $\delta'$ which we will frequently refer to later on.

Let $T$ and $T'$ be phylogenetic $X$-trees such that $T$ and $T'$ are NNI neighbors. Let $e=\{u,v\}$ be the inner edge of $T$ around which the NNI move has to  be performed in order to turn $T$ into $T'$. Then, the removal of $e$ from $T$ disconnects $T$ into two rooted subtrees. In both of these subtrees the root has degree 2 (as $e$ was an inner edge), so we can further subdivide these subtrees. Without loss of generality, we denote by $T_1$ and $T_2$ the two subtrees resulting from the deletion of $u$ and its incident edges, and similarly, we denote by $T_3$ and $T_4$ the subtrees resulting from the deletion of $v$ and its incident edges (i.e., $T$ looks like in Figure \ref{fig_nni}). The root of each $T_i$ is denoted by $\rho_i$ for $i=1,\ldots,4$.

Now let $f$ be a binary character on $X$, and let $f_i$ denote the restriction of $f$ on $T_i$ for $i=1,\ldots,4$. Then, we define  function $\delta(f,T,e)$ as follows:

\begin{equation}\label{def_delta}
\delta(f,T,e):= l(f,T)-\sum\limits_{i=1}^4 l(f_i,T_i).
\end{equation}

Informally, $\delta(f,T,e)$ describes the number of changes of $f$ on $T$ which are not contained in any of the $T_i$ ($i=1,\ldots,4$) induced by $e$, but which are rather located \enquote{around $e$} itself. We will further investigate $\delta$ subsequently and also exploit its properties to derive our main result. However, as a short-hand, whenever there is no ambiguity concerning $f$, $e$ and $T$, we often use the short-hands  $\delta=\delta(f,T,e)$ and $\delta'=\delta(f,T',e')$, where $e'$ denotes the unique edge in $T'$ that is not contained in $T$ (i.e., the new edge resulting from the NNI move, cf. Figure \ref{fig_nni}).

\subsection{Known results}

It is the main aim of the present paper to generalize the following two results concerning $k\leq 2$ to the case $k\geq 3$, at least within an NNI neighborhood of the given tree $T$.  

The first result is the following theorem which is partially based on the famous \emph{splits equivalence theorem} by Buneman \cite{Buneman1971,Semple2003}. 

\begin{theorem}[adapted from Corollary 1 and Theorem 3 of \cite{Fischer2019}]\label{thm_k<=2} 
Let $T$ be a binary phylogenetic $X$-tree with $\vert X\vert =n$. Then, $T$ is the unique maximum parsimony tree for the alignment $A_1(T)$, i.e., we have $l(A_1(T),T)<l(A_1(T),T')$ for all binary phylogenetic $X$-trees $T'\neq T$.
If additionally $n\geq 9$, then $T$ is also the unique maximum parsimony tree of alignment $A_2(T)$, i.e., we have $l(A_2(T),T)<l(A_2(T),T')$ for all binary phylogenetic $X$-trees $T'\neq T$. 
\end{theorem}

The second result we will generalize to the case $k\geq 3$ within the NNI neighborhood of $T$ is the following corollary, which does not only consider alignments of characters with parsimony score precisely $k$, but instead those of characters with parsimony score up to $k$.

\begin{corollary}[Corollary 3 in \cite{Fischer2019}] \label{concat} Let $T$ be a binary phylogenetic $X$-tree with $|X|\geq 9$. Then, $T$ is the unique maximum parsimony tree for the alignments $A_0(T).A_1(T)$, $A_0(T).A_2(T)$, $A_1(T).A_2(T)$ and $A_0(T).A_1(T).A_2(T)$, where the $.$-symbol stands for the concatenation of alignments. 
\end{corollary}

A basic result that we need throughout this paper is the following theorem\footnote{Theorem \ref{thm_Fitchrootstate} has been adapted from its original version, most prominently by adding the final  bullet point. However, it can be easily seen that this is correct as each rooted binary tree with $n$ leaves has $n-1$ inner (i.e., non-leaf) vertices. Therefore, $n-1$ is a hypothetical upper bound for the number of union vertices during the Fitch algorithm and thus also for the parsimony score. This is why a tree with $n$ leaves cannot have any characters of parsimony score more than $n-1$.}, which for a rooted binary tree counts the number of binary characters that have parsimony score $k$ and additionally assign the root a certain set $M$ during the Fitch algorithm. 

\begin{theorem}[adapted from Theorem 1 in \cite{Steel1995}] \label{thm_Fitchrootstate} Let $T$ be a rooted binary phylogenetic $X$-tree with $|X|=n$. Let $k \in \mathbb{N}_{\geq 1}$ and let $M \in \{\{a\},\{b\},\{a,b\}\}$. Let $N_M(T,k)$ denote the number of binary characters $f:X\rightarrow \{a,b\}$ on $T$ which fulfill both of the following properties:  
\begin{itemize}
\item $l(f,T)=k$ and
\item $S(f,T,\rho)=M$.
\end{itemize}

Then, we have:
\begin{itemize}
\item If $M=\{a\}$ or $M=\{b\}$ and $k\leq n-1$, we have: $$N_M(T,k)= \binom{n-k-1}{k}\cdot 2^k.$$ 
\item If $M=\{a,b\}$ and $k\leq n-1$, we have: $$N_M(T,k)=\binom{n-k-1}{k-1}\cdot 2^k.$$
\item If $k> n-1$, we have $N_M(T,k)=0$.
\end{itemize}

\end{theorem}

Equipped with all notions and known results stated above, we finally have the tools necessary to derive new results concerning the uniqueness of the MP tree $T$ of the $A_k(T)$-alignment within the NNI neighborhood of $T$.

\section{Results}

\subsection{Investigating the NNI neighborhood when \texorpdfstring{$k$}{k} is sufficiently small}
It is the main aim of this paper to show that within its NNI neighborhood and for a suitable choice of $k$, every tree $T$ gets uniquely recovered by MP from $A_k(T)$. In other words, we want to show that if $T'$ is an NNI neighbor of $T$, we have $l(A_k(T),T)<l(A_k(T),T')$ (if $k$ is sufficiently small). This result is formally stated by the following theorem.

\begin{theorem}\label{thm:NNImain}
Let $T$ and $T'$ be binary phylogenetic $X$-trees (with $\vert X \vert =n$) such that $T$ and $T'$ are NNI neighbors. Let $k\in \mathbb{N}$ such that $0<k<\frac{n}{8}+\frac{11}{9}-\frac{1}{18}\sqrt{9\cdot \left(\frac{n}{4}\right)^2+16}$. Then, we have: $l(A_k(T),T)<l(A_k(T),T')$. In other words, $T$ is the unique maximum parsimony tree of $A_k(T)$ within its NNI neighborhood.
\end{theorem}

Before we prove this theorem, we first state the following observation.

\begin{observation}\label{obs:4k+2}
The upper bound for $k$ given by Theorem \ref{thm:NNImain} is somewhat complicated. However, it can be easily shown (e.g., using a computer algebra system like Mathematica \cite{Mathematica}), that $k<\frac{n}{8}+\frac{11}{9}-\frac{1}{18}\sqrt{9\cdot \left(\frac{n}{4}\right)^2+16}$ translates to $n>9 k-11+ \sqrt{9 k^2-22 k+17} $. For all non-negative values of $k$, the latter term is strictly smaller than $12k$. Thus, the theorem in particular shows that if $n \geq 12k$, the MP tree of $A_k(T)$ is unique within the NNI neighborhood of $T$. 
\end{observation}

Next, we need to derive a technical tool to analyze the relationship between $T$ and its NNI neighbors more in-depth using $\delta$ as defined in Section \ref{sec:prelimNNI}. The following proposition will be the main ingredient in the proof of Theorem \ref{thm:NNImain}.

\begin{proposition} \label{prop:deltadiff} Let $T$ be a binary phylogenetic $X$-tree (with $\vert X \vert =n$) with an inner edge $e=\{u,v\}$ such that $T$ is as sketched in Figure \ref{fig_nni}, i.e., $e$ splits $T_1$ and $T_2$ away from $T_3$ and $T_4$. We denote the root of $T_i$ with $\rho_i$ for $i=1,\ldots, 4$. Let $T'$ be the NNI neighbor of $T$ derived by swapping $T_2$ and $T_4$ around $e$ (which generates an edge $e'$ in $T'$ corresponding to $e$ in $T$, cf. Figure \ref{fig_nni}). Let $f$ be a binary character on $X$. Then, we have: If $l(f,T)\neq l(f,T')$, then either $\delta(f,T,e)=2$ and $\delta(f,T',e')=1$ or vice versa. Moreover, we must have that two of the sets $S(\rho_1)$, $S(\rho_2)$, $S(\rho_3)$, $S(\rho_4)$ are equal to $\{a\}$ and the other two are equal to $\{b\}$, and one of the following two cases must hold: either $S(\rho_1)=S(\rho_2)\neq S(\rho_3)=S(\rho_4)$ or $S(\rho_1)=S(\rho_4)\neq S(\rho_2)=S(\rho_3)$.
\end{proposition}

\begin{proof} Let $T$ and $T'$ be two NNI neighbors as described in the statement of the proposition. Let $f$ be a binary character such that $l(f,T)\neq l(f,T')$. Note that by Equation \eqref{def_delta}, this necessarily implies $\delta:=\delta(f,T,e)\neq \delta(f,T',e')=:\delta'$.

 We now first observe that $\delta$, $\delta' \in \{0,1,2\}$. This can be easily seen by analyzing the Fitch algorithm: We consider the rooted versions $T_e$ and $T'_{e'}$ of $T$ and $T'$, respectively. These trees each have three inner vertices that are not contained in any of the subtrees $T_1,\ldots,T_4$, namely the two endpoints of $e=\{u,v\}$ and $e'=\{u',v'\}$ in the unrooted versions of the trees as well as the newly introduced root $\rho_e$ or $\rho_{e'}$, respectively. This gives a theoretical upper bound of $\delta$, $\delta'\leq 3$. Now note that in the binary case, every union node gets assigned set $\{a,b\}$ by the Fitch algorithm, so in the hypothetical case that $\delta=3$ (or $\delta'=3$), we would require three union nodes, namely $u$, $v$ and $\rho_e$ (or $u'$, $v'$ and $\rho_{e'}$, respectively). However, if the two children of the root are assigned $\{a,b\}$, the root is \emph{not} a union but an intersection node (cf. Figure \ref{fig_fitch}). Thus, we indeed must have $\delta$, $\delta' \in \{0,1,2\}$. 

 We now assume that $\delta=0$ (or $\delta'=0$). In this case, the Fitch algorithm takes no unions on any of the vertices $u$, $v$ and $\rho_e$ (or $u'$, $v'$ and $\rho_{e'}$, respectively). It can be easily seen, however, that this implies that the roots of the subtrees $T_1,\ldots,T_4$ either must all be assigned sets from $\{\{a\},\{a,b\}\}$ or they must all be assigned sets from $\{\{b\},\{a,b\}\}$ -- otherwise, inevitably at some stage an $\{a\}$-set and a $\{b\}$-set would have the same parent and thus lead to a union. However, if we have $S(\rho_i)\in \{\{a\},\{a,b\}\}$ or  $S(\rho_i)\in \{\{b\},\{a,b\}\}$ for all $i=1,\ldots,4$ (where $\rho_i$ denotes the root of $T_i$ for $i=1,\ldots,4$), then actually \emph{both} $\delta=0$ and $\delta'=0$, as -- independent of the positions of $T_1,\ldots,T_4$ in the tree -- no union will ever be formed in this case. As this would contradict $\delta\neq \delta'$, we must have $\delta>0$ and $\delta'>0$. This completes the proof of the first assertion of the proposition: As $\delta\neq \delta'$ and $\delta$, $\delta' \in \{1,2\}$, we must have that one of these values equals 1 and the other one equals 2.

In order to prove the second assertion, assume $\delta=2$. This implies that the Fitch algorithm assigns unions to two out of the three nodes $u$, $v$ and $\rho_e$. As stated above, in the binary case, unions are always $\{a,b\}$-sets, but if the parent of an $\{a,b\}$-set also gets assigned $\{a,b\}$, it is always an intersection. Thus, the two union nodes must be $u$ and $v$, and we must have $S(u)=\{a,b\}$ and $S(v)=\{a,b\}$. However, in order for these assignments to be the result of unions, we necessarily must have $S(\rho_1)\neq S(\rho_2)$ as well as $S(\rho_3)\neq S(\rho_4)$, and all sets $S(\rho_i)$ must be contained in $\{\{a\},\{b\}\}$ for $i=1,\ldots,4$. This only leaves two  possibilities: We either have $S(\rho_1)=S(\rho_3)\neq S(\rho_2)=S(\rho_4)$ or we have $S(\rho_1)=S(\rho_4)\neq S(\rho_2)=S(\rho_3)$. However, if we -- additionally to $\delta=2$ -- assume that $\delta'=1$, we cannot have $S(\rho_2)=S(\rho_4)$, because this would lead to $\delta=\delta'$, a contradiction. Thus, we necessarily must have $S(\rho_1)=S(\rho_4)\neq S(\rho_2)=S(\rho_3)$ as required. Using an analogous argument in case we have $\delta'=2$ and $\delta=1$ (in which case we must have $S(\rho_1)=S(\rho_2)\neq S(\rho_3)=S(\rho_4)$) completes the proof.
\end{proof}

\par\vspace{0.5cm}

We are now finally in the position to prove Theorem \ref{thm:NNImain}.

 \begin{proof}[Proof of Theorem \ref{thm:NNImain}]
Let $k \in \mathbb{N}$ such that $0<k<\frac{n}{8}+\frac{11}{9}-\frac{1}{18}\sqrt{9\cdot \left(\frac{n}{4}\right)^2+16}$. If $k=1$, according to Theorem \ref{thm_k<=2}, there remains nothing to show. Thus, we may assume in the following that $k\geq 2$.\footnote{Note that if $k=2$, the assumption that $k<\frac{n}{8}+\frac{11}{9}-\frac{1}{18}\sqrt{9\cdot \left(\frac{n}{4}\right)^2+16}$ ensures that $n>10$, which shows that by Theorem \ref{thm_k<=2}, also in the case $k=2$ there remains nothing to show. But as we do not require this assumption in the present proof, it is sufficient to assume $k\geq 2$ rather than $k\geq 3$ for the remainder of the proof.}

Now, let $T$ be a binary phylogenetic $X$-tree with $|X|=n$ and let $e=\{u,v\}$ be an inner edge of $T$ inducing subtrees $T_1$ and $T_2$ adjacent to $u$ and $T_3$ and $T_4$ adjacent to $v$ when $e$ is removed. In the following, we denote the roots of $T_i$ by $\rho_i$, respectively.

Without loss of generality, let $T'$ be the NNI neighbor of $T$ resulting from swapping $T_2$ and $T_4$ around $e$, leading to an inner edge $e'$ in $T'$. We need to show that $l(A_k(T),T)<l(A_k(T),T')$, or, in other words, we want to show that $l(A_k(T),T)-l(A_k(T),T')<0$.

Our proof strategy is now as follows: We analyze $A_k(T)$ and count characters of certain categories. Note that $A_k(T)$ may contain various characters that have the same parsimony score on $T$ and $T'$. These characters do not contribute to the difference of the parsimony scores and thus need not be considered as they cancel out in the above mentioned difference. So we will only consider characters $f \in A_k(T)$ for which we have $k=l(f,T)\neq l(f,T') $. By Proposition \ref{prop:deltadiff}, these characters fall into precisely two categories: 

\begin{enumerate}
\item Either $\delta(f,T,e)=2$ and $\delta(f,T',e')=1$  or
\item $\delta(f,T,e)=1$ and $\delta(f,T',e')=2$.
\end{enumerate}

In both cases, again by Proposition \ref{prop:deltadiff}, we know that two of the sets $S(\rho_1)$, $S(\rho_2)$, $S(\rho_3)$ and $ S(\rho_4)$ equal $\{a\}$ and the other two equal $\{b\}$. 
Moreover, the characters in Category (1) have $S(\rho_1)=S(\rho_4)\neq S(\rho_2)=S(\rho_3)$, whereas the characters in Category (2) have $S(\rho_1)=S(\rho_2)\neq S(\rho_3)=S(\rho_4)$.

So our proof strategy is to count the number $c_1$ of characters of Category (1) and the number $c_2$ of characters of Category (2) and show that $c_2>c_1$. This will lead to the following conclusion:

\begin{align}\label{eq_main}l(A_k(T),T)-l(A_k(T),T')&=(c_1 \cdot k + c_2 \cdot k) - (c_1 \cdot (k-1) + c_2 \cdot (k+1)) \\& = (c_1 \cdot k + c_2 \cdot k) - (c_1 \cdot k + c_2 \cdot k -c_1+c_2) \nonumber\\&=c_1-c_2 \nonumber\\&<0.\nonumber\end{align}

Here, the first equality is due to the fact that -- as explained above -- characters which do not fall into Categories (1) or (2) need not be considered in the difference. Moreover, this first equality exploits the fact that for characters in Category (1), we have $\delta(f,T',e')=1$ and thus, using  Definition \ref{def_delta} and the fact that $\sum\limits_{i=1}^4 l(f_i,T_i)=\sum\limits_{i=1}^4 l(f_i,T_i')$, $l(f,T')=k-1$ for all $f$ in $A_k(T)$ which belong to Category (1), i.e., they benefit from the NNI move from $T$ to $T'$ by a decrease in their parsimony score from $k$ to $k-1$. Similarly, for characters in Category (2), we have $l(f,T')=k+1$, i.e., such characters suffer from the NNI move from $T$ to $T'$ by an increase in their parsimony score from $k$ to $k+1$. 

\par \vspace{0.5cm}
So all that remains to be shown is that indeed we have $c_1<c_2$. 

In order to calculate $c_1$, in the light of Proposition \ref{prop:deltadiff}, as explained above we need to count the number of characters with $S(\rho_1)=S(\rho_4)\neq S(\rho_2)=S(\rho_3)$ as well as with $|S(\rho_i)|=1$ for all $i=1,\ldots, 4$. This implies that one of the following two possibilities holds:
\begin{enumerate}[(i)]
\item $S(\rho_1)=\{a\}$, $S(\rho_2)=\{b\}$, $S(\rho_3)=\{b\}$, $S(\rho_4)=\{a\}$, or 
\item $S(\rho_1)=\{b\}$, $S(\rho_2)=\{a\}$, $S(\rho_3)=\{a\}$, $S(\rho_4)=\{b\}$.
\end{enumerate}

Note that in both of these possibilities, we need to sum up over all possible values of $l(f_i,T_i)$. More precisely, we know that $l(f,T)=k$ (as $f \in A_k(T)$), and we know that $\delta(f,T,e)=2$. So by Equation \eqref{def_delta}, we must have $\sum\limits_{i=1}^4 l(f_i,T_i)=k-2$ (note that $k-2\geq 0$ is guaranteed as we assume $k\geq 2$). However, these $k-2$ changes can potentially be distributed amongst the four subtrees $T_1,\ldots,T_4$ in many possible ways, and we have to consider them all. To facilitate notation, in the following we let $k_i:=l(f_i,T_i)$ for $i=1,\ldots,4$. Then we must have $k_1+k_2+k_3+k_4=k-2$. Moreover, if we denote by $n_i$ the number of leaves in $T_i$ (for $i=1,\ldots,4$), then we must have $n_1+n_2+n_3+n_4=n$.

Furthermore, note that the number of characters with $l(f_i,T_i)=k_i$ is independent of the fact if $S(\rho_i)=\{a\}$ or $S(\rho_i)=\{b\}$. By Theorem \ref{thm_Fitchrootstate}, it simply equals $N_a(n_i,k_i):=N_{\{a\}}(T_i,k_i)=\binom{n_i-k_i-1}{k_1}\cdot 2^{k_i}$. 

Thus, in summary, we derive:

\begin{align}\label{eqc1}
c_1&=\sum\limits_{\substack{(k_1,k_2,k_3,k_4):\\k_1+k_2+k_3+k_4=k-2\\k_i\in \mathbb{N}_0}} 2 \cdot N_a(n_1,k_1)\cdot N_a(n_2,k_2)\cdot N_a(n_3,k_3) \cdot N_a(n_4,k_4),
\end{align}

where the factor 2 stems from the above two cases (i) and (ii).

\par\vspace{0.5cm}
Now in order to calculate $c_2$, again in the light of Proposition \ref{prop:deltadiff}, we need to count the number of characters with $S(\rho_1)= S(\rho_2)\neq S(\rho_3)= S(\rho_4)$ and with $|S(\rho_i)|=1$ for all $i=1,\ldots, 4$. More precisely, this means that we have two possibilities: 
\begin{enumerate}[(i)]
\item $S(\rho_1)=\{a\}$, $S(\rho_2)=\{a\}$, $S(\rho_3)=\{b\}$, $S(\rho_4)=\{b\}$, or
\item $S(\rho_1)=\{b\}$, $S(\rho_2)=\{b\}$, $S(\rho_3)=\{a\}$, $S(\rho_4)=\{a\}$.
\end{enumerate}

Note that in each of these possibilities, we need to sum up over all possible values of $l(f_i,T_i)$. More precisely, we know that $l(f,T)=k$ (as $f \in A_k(T)$), and we know that $\delta(f,T,e)=1$. So by Equation \eqref{def_delta}, we must have $\sum\limits_{i=1}^4 l(f_i,T_i)=k-1$. However, as above, these $k-1$ changes potentially can be distributed amongst the four subtrees $T_1,\ldots,T_4$ in many possible ways, and we have to consider them all. However, as this time we have $\delta(f,T,e)=1$, one of the $k_i$ needs to get incremented by 1 (compared to the $c_1$ case).

Thus, in summary, we derive:

\begin{align}\label{eqc2}
c_2=\sum\limits_{\substack{(k_1,k_2,k_3,k_4):\\k_1+k_2+k_3+k_4=k-2\\k_i\in \mathbb{N}_0}} 2\cdot & \left( N_a(n_1,k_1+1)\cdot N_a(n_2,k_2)\cdot N_a(n_3,k_3) \cdot N_a(n_4,k_4)\right.\\ &+N_a(n_1,k_1)\cdot N_a(n_2,k_2+1)\cdot N_a(n_3,k_3) \cdot N_a(n_4,k_4)\nonumber \\&+N_a(n_1,k_1)\cdot N_a(n_2,k_2)\cdot N_a(n_3,k_3+1) \cdot N_a(n_4,k_4)\nonumber \\ &\left.+N_a(n_1,k_1)\cdot N_a(n_2,k_2)\cdot N_a(n_3,k_3) \cdot N_a(n_4,k_4+1)\right).\nonumber 
\end{align}

We will now show $c_2>c_1$ by actually showing something stronger: We will show that each of the summands in the sum of $c_1$ is strictly smaller than each of the summands in the sum of $c_2$. This will complete the proof.

\par\vspace{0.5cm}
Note that we have by Theorem \ref{thm_Fitchrootstate}: 

\begin{itemize}
\item $N_a(n_i,k_i)= \binom{n_i-k_i-1}{k_i}\cdot 2^{k_i}= \frac{(n_i-k_i-1)!}{(n_i-2k_i-1)!\cdot k_i!}\cdot 2^{k_i},$
\item $N_a(n_i,k_i+1)=\binom{n_i-(k_i+1)-1}{k_i+1}\cdot 2^{k_i+1} = \binom{n_i-k_i-2}{k_i+1}\cdot 2^{k_i+1}=\frac{(n_i-k_i-2)!}{(n_i-2k_i-3)!\cdot (k_i+1)!}\cdot 2^{k_i+1}$.
\end{itemize}

\par\vspace{0.5cm}
Thus, as long as $n_i \neq k_i+1$ (so that the denominator does not equal 0), we get: \begin{align}\label{Nak+1} N_a(n_i,k_i+1)&=2 \cdot N_a(n_i,k_i)\cdot \frac{(n_i-2k_i-2)\cdot (n_i-2k_i-1)}{(n_i-k_i-1)\cdot (k_i+1)}.\end{align}
\par\vspace{0.5cm}
Moreover, we have $N_a(n_i,k_i)\geq 0$ and $N_a(n_i,k_i+1)\geq 0$ for all $n_i\in \mathbb{N}_{\geq 1}$, $k_i\in \mathbb{N}_0$ by definition. 
\par\vspace{0.5cm}
Last but not least, by the pigeonhole principle, we know that at least one of the $T_i$ has $n_i \geq \left\lceil \frac{n}{4} \right\rceil$ many leaves (because we distribute $n$ leaves of $T$ among the four subtrees $T_1, \ldots, T_4$). Without loss of generality, we therefore may assume that $n_1\geq \left\lceil \frac{n}{4} \right\rceil$. While we do not know how big $k_1$ is, we at least know that $k_1\leq k-2 < n_1-1$ since: 

\begin{align}k_1\leq k-2 &\overset{\mbox{\tiny ass. on $k$}}{<} \left(\frac{1}{2}\cdot\frac{n}{4}+\frac{11}{9}-\frac{1}{18}\sqrt{9\cdot \left(\frac{n}{4}\right)^2+16}\right)-2 \nonumber \\& \overset{\mbox{\tiny mon.incr.}}{\leq} \left(\frac{1}{2}\cdot\left\lceil\frac{n}{4}\right\rceil +\frac{11}{9}-\frac{1}{18}\sqrt{9\cdot \left(\left\lceil\frac{n}{4}\right\rceil\right)^2+16}\right)-2 \nonumber
\\& \overset{\mbox{\tiny mon.incr.}}{\leq}
\left(\frac{1}{2}\cdot n_1 +\frac{11}{9}-\frac{1}{18}\sqrt{9n_1^2+16}\right)-2 \nonumber 
\\& = \label{eq_important}
\frac{n_1}{2}-\frac{7}{9}-\frac{1}{18}\sqrt{9 n_1^2+16}\\&\leq n_1-1. \nonumber
\end{align}

Here, the first inequality is due to the theorem's assumption that $k<\frac{n}{8}+\frac{11}{9}-\frac{1}{18}\sqrt{9\cdot \left(\frac{n}{4}\right)^2+16}$. The second and third inequalities hold as $\frac{n}{4}\leq \left\lceil \frac{n}{4}\right\rceil\leq n_1$ and as the function $f(n)=\frac{1}{2}\cdot\frac{n}{4}+\frac{11}{9}-\frac{1}{18}\sqrt{9\cdot \left(\frac{n}{4}\right)^2+16}$ is monotonically increasing. The last inequality is true as it holds for all $n_1\geq 1$ (and we know that $n_i\geq 1$ for all $i=1,\ldots,4$ as each $T_i$ contains at least one leaf; otherwise $e$ would not be an inner edge).

\par\vspace{0.5cm}
So in particular, we have $n_1 \neq k_1+1$, which means we can use Equation \eqref{Nak+1} to derive $N_a(n_1,k_1+1)$ from $N_a(n_1,k_1)$.

\par\vspace{0.5cm}

Let now $(k_1,k_2,k_3,k_4)$ be a tuple such that $k_1+k_2+k_3+k_4=k-2$ and $k_i \in \mathbb{N}_0$ for all $i=1,\ldots,4$. Let $c_i^{(k_1,k_2,k_3,k_4)}$ denote the summand of $c_i$ corresponding to the tuple $(k_1,k_2,k_3,k_4)$ for $i=1,2$. Then, we have:

\begin{align*}
c_2^{(k_1,k_2,k_3,k_4)} & \overset{\eqref{eqc2}}{\geq}  2 \cdot N_a(n_1,k_1+1) \cdot N_a(n_2,k_2) \cdot N_a(n_3,k_3) \cdot N_a(n_4,k_4)\\& \overset{\eqref{Nak+1}}{=} 2 \cdot \left( 2 \cdot N_a(n_1,k_1)\cdot \frac{(n_1-2k_1-2)\cdot (n_1-2k_1-1)}{(n_1-k_1-1)\cdot (k_1+1)} \right)  \cdot N_a(n_2,k_2) \cdot N_a(n_3,k_3) \cdot N_a(n_4,k_4) \\
&= \frac{(n_1-2k_1-2)\cdot (n_1-2k_1-1)}{(n_1-k_1-1)\cdot (k_1+1)} \cdot 4 \cdot N_a(n_1,k_1) \cdot N_a(n_2,k_2) \cdot N_a(n_3,k_3) \cdot N_a(n_4,k_4) \\ &\overset{\eqref{eqc1}}{=} \frac{(n_1-2k_1-2)\cdot (n_1-2k_1-1)}{(n_1-k_1-1)\cdot (k_1+1)} \cdot 2 \cdot c_1^{(k_1,k_2,k_3,k_4)}.
\end{align*}

Thus, we definitely have $c_2^{(k_1,k_2,k_3,k_4)}>c_1^{(k_1,k_2,k_3,k_4)}$ if $ \frac{2\cdot(n_1-2k_1-2)\cdot (n_1-2k_1-1)}{(n_1-k_1-1)\cdot (k_1+1)} >1$, or, in other words, if $2\cdot(n_1-2k_1-2)\cdot (n_1-2k_1-1)>(n_1-k_1-1)\cdot (k_1+1)$. Using a computer algebra system like Mathematica \cite{Mathematica}, one can easily see that this does indeed hold in case $k_1<\frac{n_1}{2}-\frac{7}{9}-\frac{1}{18}\sqrt{9 n_1^2+16}$. However, we have already seen in Equation \eqref{eq_important} that this inequality holds.

This shows that $c_2^{(k_1,k_2,k_3,k_4)}> c_1^{(k_1,k_2,k_3,k_4)}$ for all tuples $(k_1,k_2,k_3,k_4)$ with $k_i \in \mathbb{N}_0$ for all $i=1,\ldots,4$ such that $k_1+k_2+k_3+k_4=k-2$. This in turn shows that $c_2>c_1$, which in the light of Equation \eqref{eq_main} proves that $l(A_k(T),T)>l(A_k(T),T')$ and thus completes the proof. \qedhere 
\end{proof}

\par\vspace{0.5cm}
We conclude this section with the following result, generalizing Theorem \ref{thm:NNImain} from alignments $A_k(T)$, which contain all characters with parsimony score precisely $k$ on $T$, to alignments of characters of score at most $k$ on $T$.

\begin{corollary} Let $k\in \mathbb{N}$. Let $T$ be a binary phylogenetic $X$-tree with $|X|=n>9 k-11+ \sqrt{9 k^2-22 k+17}$. Let $T'$ be an NNI-neighbor of $T$. For $i=0,\ldots,k$, let $m_i \in \mathbb{N}_0$ and let $A$ be an alignment consisting of concatenations of $A_i(T)$ such that $A_i(T)$ is contained $m_i$ times in $A$ (for $i=0,\ldots,k$). Then, $l(A,T)\leq l(A,T')$. If we additionally have $m_i>0$ for some $i>0$, then the inequality is strict.
\end{corollary}

\begin{proof}
The corollary is a direct consequence of Observation \ref{obs:4k+2} as well as Theorem \ref{thm:NNImain}, together with the observation that in case that $k=0$, the alignment $A_k(T)$ does not actually depend on $T$ -- it contains only characters that assign the same state to all elements of $X$. Thus, we have $A_0(T)=A_0(T')$ and thus also $l(A,T)=l(A,T')$ if $A$ contains only copies of $A_0(T)$ and no copies of $A_i(T)$ for $i>0$. In all other cases, the inequality is strict by Theorem \ref{thm:NNImain}.
\end{proof}

\subsection{Investigating the NNI neighborhood when \texorpdfstring{$k$}{k} is \enquote{too large}}

In this section, we want to use our findings from the previous section to derive values for $k$ for which $T$ is \emph{not} the most parsimonious tree for $A_k(T)$ in its own NNI neighborhood. In order to do so, we will exploit Equation \eqref{eq_main}. In particular, we will show that whenever $k$ assumes a certain value (depending on $n$), we can construct a tree $T$ which has an NNI neighbor $T'$ such that $l(A_k(T),T)>l(A_k(T),T')$. In other words, in these settings, $T$ is not even contained in the set of most parsimonious trees for its own $A_k$-alignment. We now state this formally in the following theorem, which will generalize Example \ref{ex:8taxa}, in which we had $n=8$ and $k=2=\frac{n-4}{2}$.

\begin{theorem}\label{thm_badcases} Let $n\in \mathbb{N}_{\geq 8}$ and let $k$ be chosen according to Table \ref{tab:badcases}. Then, there exists a binary phylogenetic $X$-tree $T$ with $|X|=n$ such that $T$ has an NNI neighbor $T'$ with $l(A_k(T),T)>l(A_k(T),T')$. In other words, $T'$ is more parsimonious for $A_k(T)$ than $T$.

\begin{table}[ht]
\begin{tabular}{|l|c|}
\hline
\multicolumn{1}{|c|}{\textbf{condition on $n$}} & \multicolumn{1}{c|}{\textbf{choice of $k$}} \\ \hline
$(n \mod 8)\equiv 0$                                  & $\frac{n-4}{2}$                             \\[5pt] \hline 
$(n \mod 8)\in\{1,7\}$                          & $\frac{n-3}{2}$                             \\[5pt] \hline
$(n \mod 8)\in\{2,6\}$                          & $\frac{n-2}{2}$                             \\[5pt] \hline
$(n \mod 8)\in\{3,5\}$                          & $\frac{n-1}{2}$                             \\[5pt] \hline
$n \mod 8 \equiv 4$                                    & $\frac{n}{2}$                               \\[5pt] \hline
\end{tabular}
\caption{Choices of $k$ for different values of $n$ to generate a \enquote{bad case} as described in Theorem \ref{thm_badcases}.}\label{tab:badcases}\end{table}
\end{theorem}

\begin{proof} We explicitly construct a tree $T$ and its NNI neighbor $T'$ with the described properties. We do this by using a strategy closely related to that of the proof of Theorem \ref{thm:NNImain}.

Let $n=8m+b$ with $m \in \mathbb{N}_{\geq 1}$ and $b \in \{0,\ldots,7\}$. We first choose four parameters $n_1$, $n_2$, $n_3$ and $n_4$ as follows:

\begin{align*} 
n_1:=\begin{cases} 2m & \mbox{if $b=0$}\\ 
2m+1 &\mbox{if $b\in \{1,\ldots , 4\}$} \\
2m+2 & \mbox{if $b\geq 5$}, \end{cases} 
&, \hspace{0.5cm}
n_2:=\begin{cases} 2m & \mbox{if $b\leq 1$}\\ 
2m+1 &\mbox{if $b\in \{2,\ldots , 5\}$} \\
2m+2 & \mbox{if $b\geq 6$}, \end{cases}
\\ 
n_3:=\begin{cases} 2m & \mbox{if $b\leq 2$}\\ 
2m+1 &\mbox{if $b\in \{3,\ldots ,6\}$} \\
2m+2 & \mbox{if $b=7$}, \end{cases}
&, \hspace{0.5cm}
n_4:=\begin{cases} 2m & \mbox{if $b\leq 3$}\\ 
2m+1 &\mbox{if $b\geq 4$}. \end{cases}
\end{align*}

Note that in all cases, we have $n_1+n_2+n_3+n_4=n$. To construct $T$ and $T'$, we now start with a single edge $e$, around which we attach rooted binary subtrees $T_1$, $T_2$, $T_3$ and $T_4$ with $n_1$, $n_2$, $n_3$ and $n_4$ leaves, respectively. Note that the particular choice of $T_i$ does not matter as long as is has $n_i$ many leaves. We make sure that in $T$, $T_1$ and $T_2$ are together on the same side of $e$, whereas $T_3$ and $T_4$ are on the other side of $e$. We then arbitrarily label the leaves such that $X=\{1,\ldots,n\}$. This is possible as $n_1+n_2+n_3+n_4=n$. We now define $T'$ as the tree resulting from $T$ by swapping $T_2$ and $T_4$. It remains to be shown that for $T$ and $T'$ with subtree sizes $n_1,\ldots,n_4$ as described above, we indeed have $l(A_k(T),T)>l(A_k(T),T')$. 

In order to see this, we now choose $k$ according to Table \ref{tab:badcases}. In particular, this leads to:

\begin{itemize} \item If $b=0$ and $n=8m+0$, we have $k=\frac{n-4}{2}=\frac{8m-4}{2}=4m-2$.
\item If $b=1$ and  $n=8m+1$, we have $k=\frac{n-3}{2}=\frac{(8m+1)-3}{2}=4m-1$.
\item  If $b=2$ and $n=8m+2$, we have $k=\frac{n-2}{2}=\frac{(8m+2)-2}{2}=4m$.
\item  If $b=3$ and  $n=8m+3$, we have $k=\frac{n-1}{2}=\frac{(8m+3)-1}{2}=4m+1$.
\item In all cases in which $b \geq 4$, we have $k=4m+2$.
\end{itemize}

From now on, we follow the strategy of the proof of  Theorem \ref{thm:NNImain}. In particular, we observe again that characters in $A_k(T)$ which have the same parsimony score on $T$ and $T'$ do not contribute to the difference $l(A_k(T),T)-l(A_k(T),T')$. We have seen in the said proof that only the $c_1$ many characters of Category 1 (for which two of the sets $S(\rho_1)$, $S(\rho_2)$, $S(\rho_3)$ and $ S(\rho_4)$ equal $\{a\}$ and the other two equal $\{b\}$ and $S(\rho_1)=S(\rho_4)\neq S(\rho_2)=S(\rho_3)$) and the $c_2$ many characters of Category 2 (for which two of the sets $S(\rho_1)$, $S(\rho_2)$, $S(\rho_3)$ and $ S(\rho_4)$ equal $\{a\}$ and the other two equal $\{b\}$ and $S(\rho_1)=S(\rho_2)\neq S(\rho_3)=S(\rho_4)$) contribute to this difference. In particular, we know $l(A_k(T),T)-l(A_k(T),T')=c_1-c_2$. 

\par \vspace{0.5cm}
We now want to calculate $c_1$ and $c_2$. We will start with showing that $c_1>0$.

In the following, we consider the rooted version $T_e$ of $T$ and denote by $\rho_i$ the root vertex of $T_i$ for $i\in \{1,\ldots,4\}$. We know by Proposition \ref{prop:deltadiff} that for all characters $f \in A_k(T)$ from Category (1), we need to make sure that $|S(\rho_i)|=1$ for all $i=1,\ldots,4$ according to the Fitch algorithm. Using Theorem \ref{thm_Fitchrootstate}, we can see that $N_a(T,k)=0$ if $k>\lceil \frac{n}{2}\rceil - \frac{1}{2}$, because then $n-k-1<\frac{n}{2}$ and $k\geq \frac{n}{2} $, so that the binomial coefficient $\binom{n-k-1}{k}$ equals 0. Therefore, we know that in order for $|S(\rho_i)|=1$ for all $i=1,\ldots,4$ to be possible, the number $k_i$ (with $i\in \{1,\ldots,4\}$)  of union nodes assigned by the Fitch algorithm within $T_i$ cannot exceed $\left\lceil\frac{n_i}{2}\right\rceil-\frac{1}{2}$. In our cases, this translates to: 

\begin{itemize} \item If $b=0$, $n=8m+0$ and $k=4m-2$, we have $n_1=n_2=n_3=n_4=2m$ and thus $k_1,k_2,k_3,k_4\leq m-1$. Note that this implies $\sum\limits_{i=1}^4 k_i\leq 4m-4= k-2$.
\item If $b=1$, $n=8m+1$ and $k=4m-1$, we have $n_1=2m+1$ and $n_2=n_3=n_4=2m$, and thus $k_1\leq m$ and $k_2,k_3,k_4\leq m-1$. Note that this implies $\sum\limits_{i=1}^4 k_i\leq 4m-3= k-2$.
\item  If $b=2$, $n=8m+2$ and $k=4m$, we have $n_1=n_2=2m+1$ and $n_3=n_4=2m$, and thus $k_1,k_2\leq m$ and $k_3,k_4\leq m-1$. Note that this implies $\sum\limits_{i=1}^4 k_i\leq 4m-2= k-2$.
\item  If $b=3$,  $n=8m+3$ and $k=4m+1$, we have $n_1=n_2=n_3=2m+1$ and $n_4=2m$, and thus $k_1,k_2,k_3\leq m$ and $k_4\leq m-1$. Note that this implies $\sum\limits_{i=1}^4 k_i\leq 4m-1= k-2$.
\item In all cases with $b \geq 4$, we have $k=4m+2$ and $k_1,k_2,k_3,k_4\leq m$. Note that this implies $\sum\limits_{i=1}^4 k_i\leq 4m= k-2$.
\end{itemize}

Now, in order for $f$ to be contained in $A_k(T)$, we require $k=l(f,T)=\sum\limits_{i=1}^4 k_i + \delta(f,T,e).$ However, we know from the above considerations that $\sum\limits_{i=1}^4 k_i\leq k-2$, and we know from the proof of Proposition \ref{prop:deltadiff} that $\delta(f,T,e) \leq 2$. Thus, in order to have $k=l(f,T)$, we need all the $k_i$ to reach their respective upper bound (cf. above bullet points).

We now finally turn our attention to $c_1$, of which we want to show that it is strictly positive. By Equation \eqref{eqc1} as well as the above considerations (showing that there is only one choice for each $k_i$, namely its upper bound according to the above bullet point list, making the summation in Equation \eqref{eqc1} redundant), we know that for our choice of $n_1$, $n_2$, $n_3$ and $n_4$, we have $c_1=2 \cdot N_a(n_1,\widehat{k}_1)\cdot N_a(n_2,\widehat{k}_2)\cdot N_a(n_3,\widehat{k}_3)\cdot N_a(n_4,\widehat{k}_4),$ where we have $\widehat{k}_i=\left\lfloor \frac{n_i-1}{2}\right\rfloor$, as this is the maximum possible integer \emph{not}  exceeding $\left\lceil \frac{n_i}{2}\right\rceil - \frac{1}{2}$. Thus, by Theorem \ref{thm_Fitchrootstate}, we know that $N_a(n_i,\widehat{k}_i)>0$ for all $i$. This immediately shows that $c_1>0$.

\par\vspace{0.5cm}
Next, we show that $c_2=0$. By Equation \eqref{eqc2} as well as the observation from above that we only have one choice for each $k_i$ (namely its respective upper bound $\widehat{k}_i$) in order to reach $k_1+k_2+k_3+k_4=k-2$ (again making the summation in Equation \eqref{eqc2} redundant), we have: 

\begin{align*}
c_2=2\cdot & \left( \underbrace{N_a(n_1,\widehat{k}_1+1)}_{=0}\cdot N_a(n_2,\widehat{k}_2)\cdot N_a(n_3,\widehat{k}_3) \cdot N_a(n_4,\widehat{k}_4)\right.\\ &+N_a(n_1,\widehat{k}_1)\cdot \underbrace{N_a(n_2,\widehat{k}_2+1)}_{=0}\cdot N_a(n_3,\widehat{k}_3) \cdot N_a(n_4,\widehat{k}_4)\nonumber \\&+N_a(n_1,\widehat{k}_1)\cdot N_a(n_2,\widehat{k}_2)\cdot \underbrace{N_a(n_3,\widehat{k}_3+1)}_{=0} \cdot N_a(n_4,\widehat{k}_4)\nonumber \\ &\left.+N_a(n_1,\widehat{k}_1)\cdot N_a(n_2,\widehat{k}_2)\cdot N_a(n_3,\widehat{k}_3) \cdot \underbrace{N_a(n_4,\widehat{k}_4+1)}_{=0}\right)\\=0.
\end{align*}

Here, the fact that $N_a(n_i,\widehat{k}_i+1)=0$ stems from $\widehat{k}_i=\left\lfloor \frac{n_i-1}{2}\right\rfloor$, which implies that $\widehat{k}_i+1$ exceeds $\left\lceil \frac{n_i}{2}\right\rceil - \frac{1}{2}$. This, by Theorem \ref{thm_Fitchrootstate}, implies $N_a(n_i,\widehat{k}_i+1)=0$.

So in summary, we have seen $c_1>0$ and $c_2=0$. As in the proof of Theorem \ref{thm:NNImain}, we can then calculate the difference of parsimony scores of $A_k(T)$ on $T$ and $T'$ as follows: 

\begin{align*} l(A_k(T),T)-l(A_k(T),T')=c_1-c_2=c_1>0, \end{align*}

which shows that $l(A_k(T),T')<l(A_k(T),T)$. This completes the proof. \qedhere

\end{proof}

\par\vspace{0.5cm}

\section{Discussion and outlook}
In this paper, we have shown that as long as $0<k<\frac{n}{8}+\frac{11}{9}-\frac{1}{18}\sqrt{9\cdot \left(\frac{n}{4}\right)^2+16}$, any binary phylogenetic tree $T$ is the unique maximum parsimony tree for $A_k(T)$ within its  NNI neighborhood. The most obvious question for future research arising from this result is if this also holds outside of this neighborhood. Note that by Theorem \ref{thm_k<=2}, for $k=1$ as well as for the combination of $k= 2$ and $n\geq 9$, this is actually the case. Moreover, the simulations presented in \cite{pablo} suggest that the result may also hold outside of the NNI neighborhood for the combinations of $k=3$ and $n\geq 12$, $k=4$ and $n\geq 15$ as well as $k=5$ and $n\geq 18$. This leads to the following conjecture, which is slightly stronger than Conjecture \ref{conj}:
\begin{conjecture}
Let $T$ be a binary phylogenetic $X$-tree wiht $|X|=n$. Let $k \leq \frac{n-3}{3}$. Then, $T$ is the unique maximum parsimony tree for $A_k(T)$.
\end{conjecture}

Another question is if $k$ really needs to be as small as required by Theorem \ref{thm:NNImain} within the NNI neighborhood. In the proof of this theorem, we made $k$ so small that not only $c_1<c_2$ as required, but that each and every summand of $c_1$ is smaller than every summand of $c_2$. Relaxing this strict condition might lead to an improved bound for $k$. We conjecture that improving $k$ is indeed possible, particularly in the light of Conjecture \ref{conj} as well as the findings presented in \cite{pablo}.

Our second result, namely the extension of Example \ref{ex:8taxa} to arbitrary values of $n\geq 8$ in Theorem \ref{thm_badcases}, in which $k$ is in the magnitude of $\frac{n}{2}$ (cf. Table \ref{tab:badcases}), is also of particular interest. In our construction of NNI neighbors $T$ and $T'$ such that $l(A_k(T),T')<l(A_k(T),T)$, in the case where $n$ is a multiple of 8 and $k=\frac{n-4}{2}$, we have $n=2k+4$. Similarly, if $n \mod 8 \equiv 1$, we have $k=\frac{n-3}{2}$ and thus $n=2k+3$. However, it was shown in \cite{Wilde2023} that for all values of $n\geq 2k+3$, the $A_k$-alignment $A_k(T)$ of $T$ is unique within the NNI neighborhood of $T$. Yet our result shows that despite this uniqueness, this does not need to imply optimality. So the question as to whether or not $T$ is the unique maximum parsimony tree of $A_k(T)$ depends not only on the question of
whether or not $T$ shares this alignment with another tree. Even in cases in which equality of scores
can be excluded, another tree might still be better than $T$ concerning the parsimony criterion -- even in the NNI neighborhood of $T$. This shows that the papers \cite{Fischer2021} and \cite{Wilde2023}, which deal with the uniqueness of the $A_k$-alignment, indeed have a different flavor than \cite{Fischer2019} as well as the present paper, which deal with the reconstructibility of $T$ from $A_k(T)$ with parsimony.

Last but not least, we hope that future research will close the gap between the case where $k$ is in the magnitude of $\frac{n}{2}$ and in which we know by Theorem \ref{thm_badcases} that the NNI neighborhood of $T$ can contain trees with a smaller parsimony score than $T$, and the case where $k$ is in the magnitude of $\frac{n}{12}$,  in which we know by Theorem \ref{thm:NNImain} and Observation \ref{obs:4k+2} that this cannot happen.

\section*{Acknowledgements} The author wants to thank Mike Steel for bringing Theorem \ref{thm_Fitchrootstate} to her attention, which helped to improve an earlier bound. Additionally, the author thanks Petula Diemke for pointing out a small error in an earlier version of the manuscript, as well as two anonynmous reviewers for their helpful comments which helped to improve the paper. Moreover, the author wishes to thank Mirko Wilde, Sophie Kersting, Linda Knüver and Stephan Dominique Andres for helpful discussions on this topic. 

\section*{Conflict of interest} The author herewith certifies that she has no affiliations with or involvement in any
organization or entity with any financial (such as honoraria; educational grants; participation in speakers’ bureaus;
membership, employment, consultancies, stock ownership, or other equity interest; and expert testimony or patent-licensing
arrangements) or non-financial (such as personal or professional relationships, affiliations, knowledge or beliefs) interest in the subject matter discussed in this paper.

\section*{Data availability statement} 
Data sharing is not applicable to this article as no new data were created or analyzed in this study.

\bibliographystyle{plain}
\bibliography{References-NEW}   

\end{document}